\documentclass[preprint,12pt,3p]{elsarticle}

\usepackage{amssymb}

\usepackage{amsthm}
\usepackage{graphicx}
\usepackage{amsmath}
\usepackage{hyperref}
\usepackage{graphicx}

\usepackage{pgf,tikz}

\usetikzlibrary{snakes}
\tikzstyle{printersafe}=[snake=snake,segment amplitude=0 pt]

\usetikzlibrary{arrows}
\usetikzlibrary{shapes}
\usepackage{multirow}
\usepackage{latexsym}
\usepackage{setspace}
\PassOptionsToPackage{boxed,section}{algorithm}
\usepackage{algorithm}
\usepackage{algorithmic}
\usepackage{enumerate}
\usepackage{amsmath}			
\usepackage{amssymb}			
\usepackage{url}
\usepackage{setspace}
\usepackage{tikz}
\usetikzlibrary{arrows}
\usetikzlibrary{shapes}
\usetikzlibrary{decorations.markings}
\usepackage{geometry}
\usepackage{bbm}
\usepackage{commath}

\newtheorem{proposition}{\em Proposition}
\newtheorem{theorem}{\em Theorem}
\newtheorem{conjecture}{\em Conjecture}
\newtheorem{definition}{\em Definition}
\newtheorem{lemma}{\em Lemma}
\newtheorem{remark}{\em Remark}
\newtheorem{corollary}{\em Corollary}
\newtheorem{claim}{\em Claim}

\journal{Sample Journal}

\begin{document}

\begin{frontmatter}

\title{On maximum $k$-edge-colorable subgraphs of bipartite graphs}

\author[label1]{Liana Karapetyan}
\address[label1]{Department of Informatics and Applied Mathematics, Yerevan State University, Yerevan, Armenia}


\ead{lianak0506@gmail.com}

\author[label5]{Vahan Mkrtchyan}
\address[label5]{Dipartimento di Informatica, University of Verona, Verona, Italy}
\ead{vahanmkrtchyan2002@ysu.am}


\begin{abstract}
If $k\geq 0$, then a $k$-edge-coloring of a graph $G$ is an assignment of colors to edges of $G$ from the set of $k$ colors, so that adjacent edges receive different colors. A $k$-edge-colorable subgraph of $G$ is maximum if it is the largest among all $k$-edge-colorable subgraphs of $G$. For a graph $G$ and $k\geq 0$, let $\nu_{k}(G)$ be the number of edges of a maximum $k$-edge-colorable subgraph of $G$. In 2010 Mkrtchyan et al. proved that if $G$ is a cubic graph, then $\nu_2(G)\leq \frac{|V|+2\nu_3(G)}{4}$. This result implies that if the cubic graph $G$ contains a perfect matching, in particular when it is bridgeless, then $\nu_2(G)\leq \frac{\nu_1(G)+\nu_3(G)}{2}$. One may wonder whether there are other interesting graph-classes, where a relation between $\nu_2(G)$ and $\frac{\nu_1(G)+\nu_3(G)}{2}$ can be proved. Related with this question, in this paper we show that $\nu_{k}(G) \geq \frac{\nu_{k-i}(G) + \nu_{k+i}(G)}{2}$ for any bipartite graph $G$, $k\geq 0$ and $i=0,1,...,k$.
\end{abstract}

\begin{keyword}
Edge-coloring; bipartite graph; $k$-edge-colorable subgraph; maximum $k$-edge-colorable subgraph.
\end{keyword}

\end{frontmatter}



\section{Introduction}

In this paper graphs are assumed to be finite, undirected and without loops, though they may contain multiple edges. The set of vertices and edges of a graph $G$ is denoted by $V(G)$ and $E(G)$, respectively. The degree of a vertex $u$ of $G$ is denoted by $d_{G}(u)$. Let $\Delta(G)$ and $\delta(G)$ be the maximum and minimum degree of a vertex of $G$. A graph $G$ is regular, if $\delta(G)=\Delta(G)$. The girth of the graph is the length of the shortest cycle in its underlying simple graph.

A {bipartite graph} is a graph whose vertices can be divided into two disjoint sets $U$ and $W$, such that every edge connects a vertex in $U$ to one in $W$. A graph is nearly bipartite, if it contains a vertex, whose removal results into a bipartite graph.

A {matching} in a graph $G$ is a subset of edges such that no vertex of $G$ is incident to two edges from the subset. A {maximum matching} is a matching that contains the largest possible number of edges.


If $k\geq 0$, then a graph $G$ is called {$k$-edge colorable}, if its edges can be assigned colors from a set of $k$ colors so that adjacent
edges receive different colors. The smallest integer $k$, such that $G$ is $k$-edge-colorable is called chromatic index of $G$ and is denoted by $\chi'(G)$. The classical theorem of Shannon states that for any graph $G$ $\Delta(G)\leq \chi'(G) \leq \left \lfloor \frac{3\Delta(G)}{2} \right \rfloor$ \cite{Shannon:1949,stiebitz:2012}. On the other hand, the classical theorem of Vizing states that for any graph $G$ $\Delta(G)\leq \chi'(G) \leq \Delta(G)+\mu(G)$ \cite{stiebitz:2012,vizing:1964}. Here $\mu(G)$ is the maximum multiplicity of an edge of $G$. A graph is class I if $\chi'(G)=\Delta(G)$, otherwise it is class II.

If the edges of $G$ are colored, then for a color $\alpha$ let $E_{\alpha}$ be the set of edges of $G$ that are colored with $\alpha$. Observe that $E_{\alpha}$ is a matching. We say that a vertex $v$ is incident to the color $\alpha$, if $v$ is incident to an edge from $E_{\alpha}$. If $v$ is not incident to the color $\alpha$, then we say that $v$ misses the color $\alpha$. Now, if we have two different colors $\alpha$ and $\beta$, then consider the subgraph of $G$ induced by $E_{\alpha}\cup E_{\beta}$. Observe that the components of this subgraph are paths or even cycles. The components which are paths are usually called $\alpha-\beta$-alternating paths or Kempe chains \cite{stiebitz:2012}. If $P$ is an $\alpha-\beta$-alternating path connecting the vertices $u$ and $v$, then we can exchange the colors on $P$ and obtain a new edge-coloring of $G$. Observe that if $u$ is incident to the color $\alpha$ in the former edge-coloring, then in the new one it will miss the color $\alpha$.

If $k<\chi'(G)$, we cannot color all edges of $G$ with $k$ colors. Thus it is reasonable to investigate the maximum number of edges that one can color with $k$ colors. A subgraph $H$ of a graph $G$ is called {maximum $k$-edge-colorable}, if $H$ is $k$-edge-colorable and contains maximum number of edges among all $k$-edge-colorable subgraphs. For $k\geq 0$ and a graph $G$ let
\[\nu_{k}(G) = \max \{ |E(H)| : H \text{ is a $k$-edge-colorable subgraph of } G \}. \]
Clearly, a $k$-edge-colorable subgraph is maximum if it contains exactly $\nu_k(G)$ edges.

There are several papers where the ratio $\frac{|E(H_k)|}{|E(G)|}$ has been investigated. Here $H_k$ is a maximum $k$-edge-colorable subgraph of $G$. \cite{bollobas:1978,henning:2007,nishizeki:1981,nishizeki:1979,weinstein:1974} prove lower bounds for the ratio when the graph is regular and $k=1$. For regular graphs of high girth the bounds are improved in \cite{flaxman:2007}. Albertson and Haas have investigated the problem in \cite{haas:1996,haas:1997} when $G$ is a cubic graph. See also \cite{samvel:2010}, where the authors proved that for every cubic graph $G$ $\nu _{2}(G) \geq \frac{4}{5}|V(G)|$ and $\nu _{3}(G) \geq  \frac{7}{6} |V(G)|$. Moreover, \cite{samvel:2014} shows that for any cubic graph $G$ $\nu _{2}(G) + \nu _{3}(G) \geq 2|V(G)|$.

    
    Bridgeless cubic graphs that are not $3$-edge-colorable are usually called snarks \cite{cavi:1998}, and the problem for snarks is investigated by Steffen in \cite{steffen:1998,steffen:2004}. This lower bound has also been investigated in the case when the graphs need not be cubic in \cite{miXumbFranciaciq:2013,Kaminski:2014,Rizzi:2009}. Kosovski and Rizzi have investigated the problem from the algorithmic perspective \cite{Kosovski:2009,Rizzi:2009}. Since the problem of constructing a $k$-edge-colorable graph in an input graph is NP-complete for each fixed $k\geq 2$, it is natural to investigate the (polynomial) approximability of the problem. In \cite{Kosovski:2009}, for each $k\geq 2$ an algorithm for the problem is presented. There for each fixed value of $k \geq 2$, algorithms are proved to have certain approximation ratios and they are tending to $1$ as $k$ tends to infinity.

Some structural properties of maximum $k$-edge-colorable subgraphs of graphs are proved in \cite{samvel:2014,MkSteffen:2012}. In particular, there it is shown that every set of disjoint cycles of a graph with $\Delta=\Delta(G) \geq 3$ can be extended to a maximum $\Delta$-edge colorable subgraph.  Also there it is shown that a maximum $\Delta$-edge colorable subgraph of a simple graph is always class I. Finally, if $G$ is a graph with girth $g \in \left \{ 2k, 2k+1 \right \} (k \geq 1)$ and $H$ is a maximum $\Delta$-edge colorable subgraph of $G$, then $\frac{|E(H)|}{|E(G)|} \geq \frac{2k}{2k+1}$ and the bound is best possible is a sense that there is an example attaining it.

In \cite{samvel:2010} Mkrtchyan et al. proved that for any cubic graph $G$ $\nu_{2}(G) \leq \frac{|V(G)| + 2\nu_{3}(G)}{4}$. For bridgeless cubic graphs, which by Petersen theorem have a perfect matching, this inequality becomes, $\nu _{2}(G)\leq \frac{\nu _{1}(G)+\nu _{3}(G)}{2}$. One may wonder whether there are other interesting graph-classes, where a relation between $\nu_2(G)$ and $\frac{\nu_1(G)+\nu_3(G)}{2}$ can be proved. In \cite{lianna:2017}, the following conjecture is stated:

\begin{conjecture}\label{conj:nearlyBip} (\cite{lianna:2017}) For each $k\geq 1$ and a nearly bipartite graph $G$
\begin{equation*}
    \nu _{k}(G) \geq \left \lfloor \frac{\nu_{k-1}(G)+\nu_{k+1}(G)}{2} \right \rfloor.
    \end{equation*} 
\end{conjecture} In the same paper the bipartite analogue of this conjecture is stated, which says that for bipartite graphs the statement of the Conjecture \ref{conj:nearlyBip} holds without the sign of floor. Note that \cite{lianna:2017} verifies Conjecture \ref{conj:nearlyBip} and its bipartite analogue when $G$ contains at most one cycle.

The present paper is organized as follows: In Section \ref{sec:aux}, some auxiliary results are stated. Section \ref{sec:main} proves the main result of the paper, which states that for any bipartite graph $G$, $k\geq 0$ $\nu_{k}(G) \geq \frac{\nu_{k-i}(G) + \nu_{k+i}(G)}{2}$, where $i=0,1,...,k$. Section \ref{sec:future} discusses the future work.

Terms and concepts that we do not define, can be found in \cite{west:1996}.

\section{Auxiliary results}
\label{sec:aux}

In this section, we present some auxiliary results that will be useful later. The first two of them are simple consequences of a classical theorem due to K\"{o}nig \cite{stiebitz:2012,west:1996}, which states for any bipartite graph $G$, we have $\chi'(G)=\Delta(G)$.

\begin{proposition}
\label{prop:Deltakprop} Let $G$ be a bipartite graph and let $k\geq 0$. Then a subgraph $F$ of $G$ is $k$-edge-colorable, if and only if $\Delta(F)\leq k$.
\end{proposition}

\begin{proposition}
\label{prop:regbip} Let $k\geq 0$ and let $G$ be a $k$-regular bipartite graph. Then for $i=0,1,...,k$ we have $\nu_i(G)=i\cdot \frac{|V(G)|}{2}$.
\end{proposition}

Our next auxiliary result follows from an observation that a vertex can be incident to at most $k$ edges in a $k$-edge-colorable graph.

\begin{proposition}
\label{prop:vertexremoval} If $G$ is a graph, $v$ is a vertex of $G$ and $k\geq 0$. Then \[\nu_k(G)\leq \nu_k(G-v)+k.\]
\end{proposition}

The next result states that if one is removing an edge from a graph, then $\nu_k(G)$ can decrease by at most one.

\begin{proposition}
\label{prop:edgeremoval} If $G$ is a graph, $e$ is an edge of $G$ and $k\geq 0$. Then \[\nu_k(G-e)\leq \nu_k(G)\leq \nu_k(G-e)+1.\]
\end{proposition}

In order to prove our next auxiliary result, we will use alternating paths.

\begin{lemma}
\label{lem:RemoveG-eHkdegrees} Let $G$ be a bipartite graph, $e=uv$ be an edge of $G$, and $j\geq 0$. Then for any maximum $j$-edge-colorable subgraph $H_j$ with $e\notin E(H_j)$, we have $d_{H_j}(u)=j$ or $d_{H_j}(v)=j$.
\end{lemma}

\begin{proof} Assume that there is a maximum $j$-edge-colorable subgraph $H_j$ that does not contain $e$ and with $d_{H_j}(u)\leq j-1$ and $d_{H_j}(v)\leq j-1$. Then there are colors $\alpha$ and $\beta$ of $H_j$ such that $\alpha$ misses at $u$ and $\beta$ misses at $v$. Clearly, $\alpha$ must be present at $v$ and $\beta$ must be present at $u$, since $H_j$ is maximum $j$-edge-colorable. Consider the $\alpha-\beta$ alternating paths starting at $u$ and $v$. If they are the same, then we get an odd cycle contradicting the fact that $G$ is bipartite. Hence they are different. Exchange the colors $\alpha$ and $\beta$ on one of them and color $e$. Observe that we have got a $j$-edge-colorable subgraph of $G$ with $|E(H_j)|+1$ edges contradicting the maximality of $H_j$. Thus the statement of the lemma should be true.
\end{proof}

If $M$ is a matching in a graph $G$, then a simple odd path $P$ is said to be $M$-augmenting, if the odd edges of $P$ lie outside $M$, the even edges of $P$ belong to $M$, and the end-points of $P$ are not covered by $M$. It is easy to see that if $G$ contains an $M$-augmenting path, then $M$ is not a maximum matching in $G$. The classical theorem of Berge \cite{berge:1973}, states that if $M$ is not a maximum matching in $G$, then $G$ must contain an $M$-augmenting path. In the end of this section, we prove the analogue of this result for $k$-edge-colorable subgraphs of bipartite graphs. It is quite plausible that our result can be derived using the general result about maximality of so-called $c$-matchings (Theorem 2 of Section 8, page 152 of \cite{berge:1973}), however, here we will give a direct proof that works only for bipartite graphs.

We will require some definitions. For a positive integer $k\geq 1$, bipartite graph $G$ and a $k$-edge-colorable subgraph $A_k$ of $G$ define an $A_k$-augmenting path as follows.

\begin{definition} A simple $u-v$-path $P$ is $A_k$-augmenting, if it is of odd length, the even edges of $P$ belong to $A_k$, the odd edges lie outside $A_k$ and $d_{A_k}(u)\leq k-1$, $d_{A_k}(v)\leq k-1$.
\end{definition}

Observe that if $G$ contains an $A_k$-augmenting path $P$, then $|E(A_k)|<\nu_k(G)$. In order to see this, consider a subgraph $B_k$ of $G$ obtained from $A_k$ by removing the even edges of $P$ from $A_k$ and adding the odd edges. Observe that any vertex $w$ of $G$ has degree at most $k$ in $B_k$, hence $B_k$ is $k$-edge-colorable by Proposition \ref{prop:Deltakprop}. Moreover, $|E(B_k)|=|E(A_k)|+1$. 
 
The following lemma states that the converse is also true.

\begin{lemma}
\label{lem:BergeAugmPathLem} Let $G$ be a bipartite graph, $k\geq 1$ and let $A_k$ be a $k$-edge-colorable subgraph with $|E(A_k)|<\nu_k(G)$. Then $G$ contains an $A_k$-augmenting path.
\end{lemma}

\begin{proof} For the $k$-edge-colorable subgraph $A_k$ consider all maximum $k$-edge-colorable subgraphs $H_k$ and choose one maximizing $|E(A_k)\cap E(H_k)|$. By an alternating component, we will mean a path or an even cycle of $G$ whose edges belong to $E(A_k)\backslash E(H_k)$ and $E(H_k)\backslash E(A_k)$, alternatively. Observe that any alternating component is either an even cycle or an even path or an odd path. Moreover, since $|E(A_k)|<\nu_k(G)$, there is at least one edge in $E(H_k)\backslash E(A_k)$, hence $G$ contains at least one alternating component.

We claim that $G$ contains no alternating component $C$ that is an even cycle. On the opposite assumption, consider a subgraph $H'_k$ of $G$ obtained from $H_k$ by exchanging the edges on $C$. Observe that the degree of any vertex of $G$ is the same as it was in $H_k$. Hence $H'_k$ is $k$-edge-colorable by Proposition \ref{prop:Deltakprop}. Moreover, $|E(H'_k)|=|E(H_k)|=\nu_k(G)$, hence $H'_k$ is maximum $k$-edge-colorable. However $|E(A_k)\cap E(H'_k)|>|E(A_k)\cap E(H_k)|$, which contradicts our choice of $H_k$.

Now, consider all alternating components $C$ of $G$ and among them choose one maximizing $|E(C)|$. From the previous paragraph we have that $C$ is a path. Let us show that $C$ is an odd path. Assume that $C$ is an even path connecting vertices $u$ and $v$. Assume that $u$ is incident to an edge of $E(H_k)\backslash E(A_k)$ and $v$ is incident to $E(A_k)\backslash E(H_k)$ on $C$. Let us show that $d_{H_k}(v)\leq k-1$. If $d_{H_k}(v)= k$, then $v$ is incident to an edge $e=vw \in E(H_k)\backslash E(A_k)$. Observe that $w\notin V(C)$. If $w\in V(C)$, then either we have an alternating component that is a cycle, or we have an odd cycle. Both of the cases are contradictory. Thus $w\notin V(C)$. Now observe that $C\cup\{e\}$ forms an alternating component with more edges than $C$. This contradicts our choice of $C$. 

Thus $d_{H_k}(v)\leq k-1$. Consider a subgraph $H'_k$ of $G$ by exchanging the edges on $C$. Observe that the degree of any vertex of $G$ is the same as it was in $H_k$ except $v$ which has degree at most $k$ and $u$ whose degree has decreased by one. Hence $H'_k$ is $k$-edge-colorable by Proposition \ref{prop:Deltakprop}. Moreover, $|E(H'_k)|=|E(H_k)|=\nu_k(G)$, hence $H'_k$ is maximum $k$-edge-colorable. However $|E(A_k)\cap E(H'_k)|>|E(A_k)\cap E(H_k)|$, which contradicts our choice of $H_k$.

Thus $C$ is an odd path. Again let the end-points of $C$ be $u$ and $v$. If $u$ and $v$ are incident to edges $E(A_k)\backslash E(H_k)$ on $C$, then similarly to previous paragraph, one can show that $d_{H_k}(u)\leq k-1$ and $d_{H_k}(v)\leq k-1$. If we exchange the edges of $H_k$ on $C$ we would find a larger $k$-edge-colorable subgraph, contradicting the maximality of $H_k$.

Thus, $u$ and $v$ are incident to edges $E(H_k)\backslash E(A_k)$ on $C$. Similarly to previous paragraph, one can show that $d_{A_k}(u)\leq k-1$ and $d_{A_k}(v)\leq k-1$. Now, it is not hard to see that $C$ is an $A_k$-augmenting path. The proof of the lemma is complete.
\end{proof}

\bigskip

When $G$ is not bipartite, $G$ may possess an augmenting path with respect to a maximum $k$-edge-colorable subgraph. Consider the graph from Figure \ref{fig:AugNonBipExample}, and let $A_2$ be the subgraph colored with $\alpha$ and $\beta$. It is easy to see that $A_2$ is maximum $2$-edge-colorable in $G$, however $G$ contains an $A_2$-augmenting path.

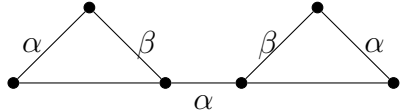
\begin{figure}[htbp]
 \begin{center}
  \begin{tikzpicture}
  
  \node at (0.25,0.5) {$\alpha$};
  \node at (1.75,0.5) {$\beta$};
  \node at (2.5,-0.25) {$\alpha$};
  \node at (3.35,0.5) {$\beta$};
  \node at (4.75,0.5) {$\alpha$};
  
  \tikzstyle{every node}=[circle, draw, fill=black!50,
                        inner sep=0pt, minimum width=4pt]
                        
    \node[circle,fill=black,draw] at (0,0) (n00) {};
    \node[circle,fill=black,draw] at (1,1) (n11) {};
    \node[circle,fill=black,draw] at (2,0) (n20) {};
    \node[circle,fill=black,draw] at (3,0) (n30) {};
    \node[circle,fill=black,draw] at (4,1) (n41) {};
    \node[circle,fill=black,draw] at (5,0) (n50) {};
    
    \path[every node]
            
            (n00) edge  (n11)
                  edge (n20)
                  
            (n11) edge (n20)
            
            (n20) edge (n30)
            
            (n30) edge (n41)
                  edge (n50)
                  
            (n41) edge (n50);
   \end{tikzpicture}
 \end{center}
\caption{The statement of Lemma \ref{lem:BergeAugmPathLem} is not true when $G$ is not bipartite.}
\label{fig:AugNonBipExample}       
\end{figure}

\section{The main results}
\label{sec:main}

In this section, we obtain the main result of the paper. Our first theorem proves a lower bound for $\nu_k(G)$ in terms of the average of $\nu_{k-1}(G)$ and $\nu_{k+1}(G)$.

\begin{theorem}\label{thm:Bipk=k} For any bipartite graph $G$ and $k\geq 1$
\[	\nu_k(G)\geq \frac{\nu_{k-1}(G) + \nu_{k+1}(G)}{2}.\]
\end{theorem}

\begin{proof} Assume that the statement of the theorem is wrong. Let $G$ be a counter-example minimizing $|V(G)|+|E(G)|$. We prove a series of claims that establish various properties of $G$.

\begin{claim}\label{claim:simple} $G$ is connected and $|V(G)|\geq 2$.
\end{claim}

\begin{proof} If $G$ is the graph with one vertex, then clearly it is bipartite and $\nu_i(G)=0$ for any $i\geq 0$, hence it is not a counter-example to our theorem. Thus, $|V(G)|\geq 2$. Let us show that $G$ is connected. Assume that $G$ contains $t\geq 2$ components, which are $G^{(1)},...,G^{(t)}$. We have that for $i\geq 0$
\[\nu_i(G)=\nu_i(G^{(1)})+...+\nu_i(G^{(t)}),\]
hence
\begin{align*}
\nu_k(G)&=\nu_k(G^{(1)})+...+\nu_k(G^{(t)})\geq \frac{\nu_{k-1}(G^{(1)}) + \nu_{k+1}(G^{(1)})}{2}+...+\frac{\nu_{k-1}(G^{(t)}) + \nu_{k+1}(G^{(t)})}{2}\\
&=\frac{\nu_{k-1}(G) + \nu_{k+1}(G)}{2}.
\end{align*} Thus, $G$ is not a counter-example to our statement contradicting our assumption. Here we used the fact that $G^{(1)},...,G^{(t)}$ are smaller than $G$, hence they are not counter-examples to our theorem.
The proof of the claim is complete.
\end{proof}

\begin{claim}
\label{cl:Hk-1Hk+1} For any maximum $(k-1)$-edge-colorable subgraph $H_{k-1}$ and any maximum $(k+1)$-edge-colorable subgraph $H_{k+1}$, we have 
\[E(H_{k-1})\cup E(H_{k+1})=E(G).\]
\end{claim}

\begin{proof} If $ E(H_{k-1})\cup E(H_{k+1}) \neq E(G)$ for some $H_{k-1}$ and $H_{k+1}$, then there exist an edge $e$, such that $e$ lies outside $H_{k-1}$ and $H_{k+1}$. Hence
\[\nu_{k-1}(G-e)=\nu_{k-1}(G)\] and \[\nu_{k+1}(G-e)=\nu_{k+1}(G),\]
therefore we get:
\[\nu_k(G)\geq \nu_k(G-e)\geq \frac{\nu_{k-1}(G-e) + \nu_{k+1}(G-e)}{2}=\frac{\nu_{k-1}(G) + \nu_{k+1}(G)}{2}.\]
Here we used the fact that the bipartite graph $G-e$ is not a counter-example.
\end{proof}

Our next claim states that removing an edge from $G$ does not decrease the size of $\nu_k(G)$.

\begin{claim}
\label{cl:RemoveG-e} For any edge $e$ of $G$, we have $\nu_k(G)=\nu_k(G-e)$.
\end{claim}

\begin{proof} If $\nu_k(G)=1+\nu_k(G-e)$ (Proposition \ref{prop:edgeremoval}), then 
\begin{align*}
    \nu_{k}(G) &= 1+\nu_k(G-e) \geq 1+  \frac{\nu_{k-1}(G-e)  + \nu_{k+1}(G-e) }{2}\\
               & =  \frac{\nu_{k-1}(G-e)+1  + \nu_{k+1}(G-e)+1 }{2}\geq  \frac{\nu_{k-1}(G) + \nu_{k+1}(G)}{2}.
\end{align*} Here we used the fact that $G-e$ is not a counter-example and Proposition \ref{prop:edgeremoval} twice.
\end{proof}

Our final claim establishes some relations for maximum and minimum degrees of $G$. Its proof makes use of the fan-argument by Vizing \cite{stiebitz:2012,vizing:1964}.

\begin{claim}
\label{cl:mimmaxdegree} $\Delta(G)\leq 2k$ and $\delta(G)\leq k$.
\end{claim}

\begin{proof} Let $H_{k-1}$ and $H_{k+1}$ be a maximum $(k-1)$-edge-colorable and a maximum $(k+1)$-edge-colorable subgraphs of $G$, respectively. By Claim \ref{cl:Hk-1Hk+1} $G$ is a union of $H_{k-1}$ and $H_{k+1}$, hence it is a union of $2k$ matchings. Thus $\Delta(G)\leq 2k$. 

Let us show that $\delta(G)\leq k$. Assume that $\delta(G)\geq k+1$. If $\Delta(G)\leq k+1$, then $G$ is $(k+1)$-regular, hence from Proposition \ref{prop:regbip} we have  $\nu_{i}(G)=i\cdot \frac{|V|}{2}$ for $i=k-1,k,k+1$. Therefore
\[\nu_k(G)= \frac{\nu_{k-1}(G) + \nu_{k+1}(G)}{2}.\]
Thus, $G$ is not a counter-example. Hence, we can assume that $\Delta(G)\geq k+2$, and therefore $E(H_{k-1})\backslash E(H_{k+1})\neq \emptyset$. Let $e=uv$ be an edge from this set. Then $u$ or $v$ must be incident to all $(k+1)$ colors of $H_{k+1}$ (apply Lemma \ref{lem:RemoveG-eHkdegrees} with $j=k+1$). Assume that this vertex is $v$. Let us show that $u$ is incident to all $(k+1)$ colors of $H_{k+1}$ as well.

On the opposite assumption, assume that $u$ misses a color $\beta$ of $H_{k+1}$. Then $v$ must be incident to an edge $e_w=vw$ of color $\beta$ in $H_{k+1}$, as $d_{H_{k+1}}(v)= k+1$. Since $d_{H_{k-1}}(v)\leq k-1$ and $d_{H_{k+1}}(v)= k+1$, there is an edge $e_z=vz$ incident to $v$ such that $e_z\in E(H_{k+1})\backslash E(H_{k-1})$. Let the color of $e_z$ in $H_{k+1}$ be $\alpha$.

If $\alpha$ is missing at $u$, then consider a subgraph $H'_{k+1}$ of $G$ obtained from $H_{k+1}$ by removing the edge $e_z$, adding $e$ to $H'_{k+1}$ and coloring $e$ with $\alpha$. Observe that $H'_{k+1}$ is $(k+1)$-edge-colorable, $|E(H'_{k+1})|=|E(H_{k+1})|=\nu_{k+1}(G)$. Hence $H'_{k+1}$ is maximum $(k+1)$-edge-colorable. However, $e_z\notin E(H_{k-1})\cup E(H'_{k+1})$ violating Claim \ref{cl:Hk-1Hk+1}.

Thus, we can assume that $\alpha$ is present at $u$, hence it is different from $\beta$. Consider the $\alpha-\beta$ alternating path $P_u$ of $H_{k+1}$ starting from $u$. We claim that $P_u$ passes through $v$. If not, we could have exchanged the colors on $P_u$, remove $e_z$ from $H_{k+1}$, add $e$ to $H_{k+1}$, color it with $\alpha$ and get a new maximum $(k+1)$-edge-colorable subgraph violating Claim \ref{cl:Hk-1Hk+1}. Thus, $P_u$ passes through $v$. We claim that it passes first via $z$, then via $v$ and $w$. If $P_u$ first passes via $w$, then together with $e$ we get an odd cycle contradicting our assumption. 

Let $P_w$ be the final part of $P_u$ that starts from $w$. Consider a $(k+1)$-edge-colorable subgraph $H'_{k+1}$ of $G$ obtained from $H_{k+1}$ as follows: exchange the colors on $P_w$, color $e$ with $\beta$, color $e_w$ with $\alpha$ and remove $e_z$ from $H_{k+1}$. Observe that $H'_{k+1}$ is $(k+1)$-edge-colorable, $|E(H'_{k+1})|=|E(H_{k+1})|=\nu_{k+1}(G)$. Hence $H'_{k+1}$ is maximum $(k+1)$-edge-colorable. However, $e_z\notin E(H_{k-1}) \cup E(H'_{k+1})$ violating Claim \ref{cl:Hk-1Hk+1}.

Thus $u$ and $v$ must be incident to all $(k+1)$ colors of $H_{k+1}$, in particular, $d(u)\geq k+2$ and $d(v)\geq k+2$. Observe that by Claim \ref{cl:Hk-1Hk+1}, any vertex of degree at least $k+2$ must be incident to an edge from $E(H_{k-1})\backslash E(H_{k+1})$. Consider the bipartite graph $J=G-(E(H_{k-1})\backslash E(H_{k+1}))$. Observe that $J$ is a $(k+1)$-regular bipartite graph with $V(J)=V(G)$. Hence from Proposition \ref{prop:regbip}, we have $\nu_{i}(G)=i\cdot \frac{|V|}{2}$ for $i=k-1,k,k+1$, and therefore
\[\nu_k(G)= \frac{\nu_{k-1}(G) + \nu_{k+1}(G)}{2}.\]
This means that $G$ is not a counter-example to our statement contradicting our assumption. Hence $\delta(G)\leq k$. The proof of the claim is complete.
\end{proof}

We are ready to prove the theorem. By Claim \ref{cl:mimmaxdegree}, $\delta(G)\leq k$, hence there is a vertex $u$ with $d_G(u)\leq k$. On the other hand, by Claim \ref{claim:simple} $G$ is connected and $|V|\geq 2$, hence $d_G(u)\geq 1$. Thus, there is an edge $e=uv$ incident to $u$. By Claim \ref{cl:RemoveG-e}, there is a maximum $k$-edge-colorable subgraph $H_k$ that does not contain $e$. By Lemma \ref{lem:RemoveG-eHkdegrees}, $d_{H_k}(u)=k$ or $d_{H_k}(v)=k$ for any such $H_k$. Since $d_{G-e}(u)\leq k-1$, we have $d_{H_k}(v)=k$ for any maximum $k$-edge-colorable subgraph $H_k$ that does not contain $e$.

By Proposition \ref{prop:vertexremoval}, we have $\nu_k(G)\leq \nu_k(G-v)+k$. Let us show that $\nu_k(G)=\nu_k(G-v)+k$. Assume that $\nu_k(G)\leq \nu_k(G-v)+k-1$. Since $\nu_k(G)=\nu_k(G-e)$ (Claim \ref{cl:RemoveG-e}) and $G-e-v=G-v$, we have $\nu_k(G-e)\leq \nu_k(G-e-v)+k-1$.

Choose a maximum $k$-edge-colorable subgraph $H^{(0)}$ of $G-e-v$. If $H^{(0)}$ is maximum in $G-e$, then since $e$ does not lie in $H^{(0)}$, we have a contradiction with $d_{H^{(0)}}(v)=k$ as $d_{H^{(0)}}(v)=0$. Thus $H^{(0)}$ is not maximum in $G-e$. By Lemma \ref{lem:BergeAugmPathLem}, there a $k$-edge-colorable subgraph $H^{(1)}$ which is obtained from $H^{(0)}$ by shifting the edges on an $H^{(0)}$-augmenting path in $G-e$. Observe that $d_{H^{(1)}}(v)\leq 1$. If $H^{(1)}$ is maximum in $G-e$, then we have a contradiction with $d_{H^{(1)}}(v)=k$ as $d_{H^{(1)}}(v)\leq 1$. Thus $H^{(1)}$ is not maximum in $G-e$. By repeating the argument and applying Lemma \ref{lem:BergeAugmPathLem} at most $(k-1)$ times, we will obtain a maximum $k$-edge-colorable subgraph $H^{(i)}$ of $G-e$ with $d_{H^{(i)}}(v)\leq k-1$ contradicting the fact that $d_{H_k}(v)=k$ for any maximum $k$-edge-colorable subgraph $H_k$ of $G$ that does not contain $e$.

Thus, $\nu_k(G)=\nu_k(G-v)+k$. We have
\begin{align*}
    \nu_{k}(G) &= k+\nu_k(G-v) \geq k+  \frac{\nu_{k-1}(G-v)  + \nu_{k+1}(G-v) }{2}\\
               & =  \frac{\nu_{k-1}(G-v)+k-1  + \nu_{k+1}(G-v)+k+1 }{2}\geq  \frac{\nu_{k-1}(G) + \nu_{k+1}(G)}{2}.
\end{align*} This contradicts the fact that $G$ is a counter-example to our statement. Here we used the fact that $G-v$ is not a counter-example and Proposition \ref{prop:vertexremoval} twice. The proof of the theorem is complete.
\end{proof}

The proved theorem is equivalent to the following

\begin{remark}\label{rem:decrremark} If $G$ is a bipartite graph, then 
\[\nu_1(G)-\nu_0(G)\geq \nu_2(G)-\nu_1(G)\geq \nu_3(G)-\nu_2(G)\geq ... .\]
\end{remark}

Below we derive the main result of the paper as a corollary to the theorem proved above:

\begin{corollary}
\label{cor:ArithMeanCorollary} Let $G$ be a bipartite graph and let $k\geq 0$. Then for $i=0,1,...,k$ we have
\[\nu_k(G)\geq \frac{\nu_{k-i}(G)+\nu_{k+i}(G)}{2}.\]
\end{corollary}

\begin{proof} We prove the statement by induction on $i$. When $i=0$, the statement is trivial. When $i=1$, it follows from Theorem \ref{thm:Bipk=k}. We will assume that the statement is true for $i-1$, and prove it for $i$.

By induction hypothesis we have
\[\nu_k(G)\geq \frac{\nu_{k-i+1}(G)+\nu_{k+i-1}(G)}{2}.\]
By applying Theorem \ref{thm:Bipk=k} on $\nu_{k-i+1}(G)$ and $\nu_{k+i-1}(G)$ we have
\[  \nu_{k}(G) \geq\frac{ \nu_{k-i+2}(G) + \nu_{k-i}(G) + \nu_{k+i-2}(G) + \nu_{k+i}(G)}{4} = 
\frac{\nu_{k-i}(G) +  \nu_{k+i}(G)}{4} + \frac{ \nu_{k-i+2}(G) + \nu_{k+i-2}(G) }{4}.\]
So, in order to complete the proof of the corollary, we need to show
\begin{equation}\label{eq:ToBeProved}
\nu_{k-i+2}(G) + \nu_{k+i-2}(G) \geq \nu_{k-i}(G) + \nu_{k+i}(G).
\end{equation}

Using Remark \ref{rem:decrremark}, we have
\[  \nu_{k-i+1}(G) - \nu_{k-i}(G) \geq \nu_{k-i+2}(G) - \nu_{k-i+1}(G)\geq \cdots \geq \nu_{k+i-1}(G) - \nu_{k+i-2}(G) \geq \nu_{k+i}(G) - \nu_{k+i-1}(G).\]
The last inequality implies
\[[\nu_{k-i+1}(G) - \nu_{k-i}(G)] + [\nu_{k-i+2}(G) - \nu_{k-i+1}(G)] \geq [\nu_{k+i-1}(G) - \nu_{k+i-2}(G)] + [\nu_{k+i}(G) - \nu_{k+i-1}(G)],\]
or
\[  \nu_{k-i+2}(G) - \nu_{k-i}(G) \geq  \nu_{k+i}(G) - \nu_{k+i-2}(G),  \]
which is equivalent to (\ref{eq:ToBeProved}). The proof of the corollary is complete.
\end{proof}

\section{Future Work}
\label{sec:future}

For a (not necessarily bipartite) graph $G$, let $b(G)$ be the smallest number of vertices of $G$ whose removal results into a bipartite graph. One can easily see that $b(G)$ coincides with the minimum number of vertices of $G$, such that any odd cycle of $G$ contains a vertex from these vertices. $b(G)$ is a well studied parameter frequently appearing in various papers on Graph theory and Algorithms. It can be easily seen that a graph $G$ is bipartite if and only if $b(G)=0$, and is nearly bipartite if and only if $b(G)\leq 1$.

We suspect that:

\begin{conjecture}\label{conj:b(G)conjecture} Let $G$ be a graph and let $k\geq 0$. Then for $i=0,1,...,k$ we have
\[\nu_k(G)\geq \frac{\nu_{k-i}(G)+\nu_{k+i}(G)-b(G)}{2}.\]
\end{conjecture} Observe that when $G$ is bipartite, we get the statement of Corollary \ref{cor:ArithMeanCorollary}. On the other hand, when $G$ is nearly bipartite and $i=1$, we get the statement of the Conjecture \ref{conj:nearlyBip}.

\section*{Acknowledgement}

The second author is indebted to Armen Asratian for useful discussions on $c$-matchings and for disproving an early version of Lemma \ref{lem:BergeAugmPathLem}.




\nocite{*}


\bibliographystyle{abbrv}


\end{document}